\documentclass[10pt,conference]{IEEEtran}

\usepackage{amssymb}
\usepackage{amsmath}
\usepackage{cite}
\usepackage{url}
\usepackage{xcolor}
\usepackage{cite,graphicx,amsmath,amssymb}
\usepackage{fancyhdr}
\usepackage{mdwmath}
\usepackage{mdwtab}
\usepackage{caption}
\usepackage{amsthm}
\usepackage{bm}
\usepackage{amsmath}   
\usepackage{amssymb}   
\usepackage{mathrsfs}
\usepackage{setspace}
\usepackage{algorithm}
\usepackage{algorithmic}
\usepackage{pdfpages}
\usepackage{mathtools}
\usepackage{subcaption}
\usepackage{bbm}

\graphicspath{{./src/img/}}

\newtheorem{remark}{Remark}
\newtheorem{theorem}{Theorem}

\newtheorem{lemma}{Lemma}
\newtheorem{assumption}{Assumption}

\newtheorem{corollary}{Corollary}

\newtheorem{proposition}{Proposition}

\allowdisplaybreaks
\expandafter\def\expandafter\normalsize\expandafter{%
    \normalsize%
    \setlength\abovedisplayskip{4pt}%
    \setlength\belowdisplayskip{4pt}%
    \setlength\abovedisplayshortskip{2pt}%
    \setlength\belowdisplayshortskip{2pt}%
}

\title{Inference-Optimal ISAC via Task-Oriented Feature Transmission and Power Allocation}
\author{
    \IEEEauthorblockN{Biao Dong, Bin Cao, and Qinyu Zhang\\
\IEEEauthorblockA{Harbin Institute of Technology, Shenzhen, China
\\E-mail: 23b952012@stu.hit.edu.cn, caobin@hit.edu.cn, and zqy@hit.edu.cn}}
\vspace{-0.7cm}
}

\begin{document}

\maketitle
\begin{abstract}
   This work is concerned with the coordination gain in integrated sensing and communication (ISAC) systems under a compress-and-estimate (CE) framework, wherein inference performance is leveraged as the key metric. To enable tractable transceiver design and resource optimization, we characterize inference performance via an error probability bound as a monotonic function of the discriminant gain (DG). This raises the natural question of whether maximizing DG, rather than minimizing mean squared error (MSE), can yield better inference performance. Closed-form solutions for DG-optimal and MSE-optimal transceiver designs are derived, revealing water-filling-type structures and explicit sensing and communication (S\&C) tradeoff. Numerical experiments confirm that DG-optimal design achieves more power-efficient transmission, especially in the low signal-to-noise ratio (SNR) regime, by selectively allocating power to informative features and thus saving transmit power for sensing.   
\end{abstract}

\section{Introduction}
Integrated sensing and communication (ISAC) has been recognized as a key enabling technology for 6G, owing to its \emph{integration gain} and \emph{coordination gain}. The well-known integration gain allows sensing and communication (S\&C) to share radio spectrum, hardware infrastructure, and signal processing modules within a unified platform \cite{liu2022integrated}. In contrast, the \emph{coordination gain} has received comparatively less attention \cite{dong2025communication}. As illustrated in Fig.~\ref{2intro_Project4_Task_ISCC}, two dominant factors determine the overall system performance: \emph{sensing noise} and \emph{fading noise}. Since the information transmitted over the communication link originates from the sensing process, efficient system resource management, including joint optimization of power, bandwidth, and waveform design, is crucial for achieving optimal cooperation between sensing and communication and thus maximizing the coordination gain.


Most existing studies address specific parameter estimation tasks (e.g., angle or velocity) and adopt estimate-and-compress (EC) based schemes \cite{dong2025communication,luo2005universal,berger1996ceo}, where local ISAC nodes estimate parameters and then transmit them. For more complex inference tasks \footnote{Also known as pattern recognition, it determines the category of a given sample and represents a task-oriented classification process.}, local ISAC devices often lack sufficient computational capability. In such cases, edge inference becomes more practical \cite{wen2023task}, i.e., offloading computation-intensive tasks to edge servers, thereby enhancing device capability and extending battery life. This approach corresponds to the compress-and-estimate (CE) based scheme.

Under the CE scheme, the key problem lies in designing the S\&C transceiver to maximize inference performance. However, directly optimizing inference performance is generally intractable due to the lack of an explicit analytical expression for the inference metric. Consequently, existing works mainly resort to surrogate or qualitative metrics. For example, \cite{lan2022progressive} proposed a tractable surrogate metric called discriminant gain (DG), which is derived under the assumption that each class follows a Gaussian distribution. Essentially, DG corresponds to the Mahalanobis distance between two classes. \cite{wen2023task} incorporated sensing noise into the DG and derived a sensing-noise-aware beamforming design. \cite{chen2024view} and \cite{wang2025ultra} further derived performance bounds based on DG. 



Different from above-mentioned, this paper investigates inference-oriented S\&C transceiver design and resource optimization under the CE scheme, with the objective of maximizing inference performance. We revisit the improvement of inference accuracy from a fundamental perspective, providing a theoretical analysis that explicitly links DG, power allocation, and system noise to the achievable inference gain. The analysis is based on the inference error probability bound, which is a monotonic function of DG. We reveal that the DG-optimal transceiver design is power-efficient compared with the conventional minimum mean square error (MMSE) criterion, as its gain arise from the power allocation strategy across heterogeneous feature dimensions, thereby saving power for sensing. 

{\textit Notations}: We use boldface letters for vectors (e.g., $\boldsymbol{r}$, $\boldsymbol{g}$), $(\cdot)^\top$ for transpose, $\mathcal{CN}(a,b^2)$ for the circularly symmetric complex Gaussian (CSCG) distribution with mean $a$ and variance $b^2$, $\mathbf{I}_{a \times a}$ for the $a \times a$ identity matrix, and $\mathrm{diag}(\cdot)$ for the diagonal matrix formed from its arguments.

\begin{figure}[!t]
    \centering
    \includegraphics[width=0.4\textwidth]{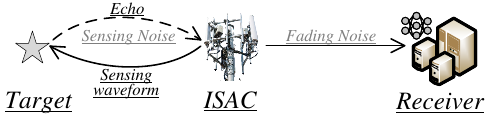}
    \caption{The considered S\&C system, where one ISAC device transmits probing signals to detect the target and forwards the compressed features of the sensed data to the receiver for inference.}
    \label{2intro_Project4_Task_ISCC}
\end{figure} 

\section{System Model} \label{sec:system_model}
As shown in Fig.~\ref{2intro_Project4_Task_ISCC}, we consider an S\&C system consisting of a single ISAC device, a target, and a receiver. Following \cite{li2023integrated,liu2018mu}, each ISAC device is equipped with three antennas in a separated deployment way, where two antennas are dedicated to target sensing and reception respectively \cite[Chapter 2.2]{jankiraman2018fmcw}, and one for communication. The signal transmitted can be expressed as
 \begin{equation}\label{Eq:signalTransmitted}
u(t) = \begin{bmatrix}
\sqrt{P_{\mathsf r}}\,{u}_{\mathsf r}[t] \\
\sqrt{P_{\mathsf c}}\,{u}_{\mathsf c}[t]
\end{bmatrix},
\end{equation}
where $P_{\mathsf r}$ and $P_{\mathsf c}$ denote the transmit power allocated to the sensing signal ${u}_{\mathsf r}[t]$ and the communication signal ${u}_{\mathsf c}[t]$, respectively, subject to a total power constraint as
\begin{equation}\label{Eq:P_r_P_c}
   P_{\mathsf r}+P_{\mathsf c} \leq P. 
\end{equation}
\subsection{Sensing Model}
\subsubsection{Sensing Signal Processing}A sensing snapshot consists of $C$ chirps, each with a duration of $T_{o}$. The sensing waveform within one snapshot, denoted as ${u}_{\mathsf r}[t]$, employs frequency-modulated continuous-wave (FMCW) modulation for target sensing, and can be expressed as
\begin{align*}
\!\!\!\!\!u_{\mathsf r}[t] 
\!\!\!=\!\!\! \sum_{c=0}^{C-1} 
   \!\!\text{rect}\!(\frac{t\!-\!cT_{o}}{T_{o}})\! 
   \cos( 
      2\pi f_{c,k}\,(t\!-\!cT_{o}) 
      \!\!+\!\! 2\pi \frac{B_{s}}{T_{o}} (t\!-\!cT_{o})^{2}),
\end{align*}
where $\text{rect}(\cdot)$ denotes a rectangular pulse function of unit width centered at $t = 0$, $f_{c}$ represents the sensing carrier frequency, and $B_{s}$ denotes the sensing bandwidth. The echo can be modeled as
\begin{equation}\label{echo}
   \!\!\!\!{r}(t) \!=\! \underbrace{\sqrt{P_{\mathsf r}}\, \alpha e^{j2\pi f_{D} t} u_{\mathsf r}(t-\tau)}_{\text{Desired echo signal}} 
+ \underbrace{v_{\mathsf r}(t) + v_{\mathsf c}(t) + n(t),}_{\text{Equivalent noise}}
\end{equation}
where $\alpha$ denotes the attenuation factor of the target round-trip, $\tau$ is the round-trip delay, $f_{D}$ denotes the Doppler shift, and $n(t)$ is the additive white Gaussian noise (AWGN). The term $v_{\mathsf r}(t)$ is the non-target clutter, which can be mitigated via pulse compression and hence approximated as an equivalent Gaussian noise by the law of large-number \cite[Eq. 54]{li2023integrated}. The term $v_{\mathsf c}(t)$ denotes the communication self-interference component. 

We assume that ${u}_{\mathsf r}[t]$ is statistically independent of ${u}_{\mathsf c}[t]$, i.e., $\mathbb{E}\{u_{\mathsf r}[t]\, u_{\mathsf c}^*[t]\} = 0$.
Based on this assumption and pulse compression, $v_{\mathsf c}(t)$ can be approximated as an equivalent Gaussian noise. Specifically, applying the matched filtering $h_{\text{MF}}(t) = u_{\mathsf r}^*(-t)$ to \eqref{echo}, the received signal is expressed as \cite{wen2023task,li2023integrated}
\begin{equation}\label{Eq:ReceivedSignalMatchedFilter}
    \bar r(t) = \underbrace{\sqrt{P_{\mathsf r}}\,\alpha e^{j2\pi f_{D} t}\, R_{\text{AU}}(t-\tau)}_{\text{Pulse-compressed target echo}}
+\!\!\!\!\!\!\! \underbrace{\tilde{n}(t)}_{\text{Equivalent noise}}\!\!\!\!\!\!,
\end{equation}where $R_{\text{AU}}(\cdot)$ denotes the autocorrelation function of the radar probing waveform, and the equivalent noise term is given by
$\tilde{n}(t) = \big(v_{\mathsf r}(t) + v_{\mathsf c}(t) + n(t)\big) \ast u_{\mathsf r}^*(-t)$, which follows $\mathcal{CN}(0,\sigma_{\mathsf r}^2)$.
The sampled radar echo is then reorganized into a sensing matrix $\bar{\mathbf{R}} \in \mathbb{C}^{f_sT_o \times C}$, where each column contains the fast-time samples acquired at rate $f_s$ (corresponding to target range) and each row represents slow-time snapshots across chirps, capturing Doppler frequency shifts \cite{wen2023task}. Vectorizing $\bar{\mathbf{R}}$ yields the slow-time sequence $\bar{\boldsymbol{r}} \in \mathbb{C}^{f_sT_o C}$. To capture the time-varying Doppler characteristics, $\boldsymbol{\bar r}$ is further processed via short-time Fourier transform (STFT). Let $f \in \mathcal{F}=\{0,1,\dots,F-1\}$ denote the frequency index and $t \in \mathcal{T}=\{0,1,\dots,T_{\rm stft}-1\}$ the time-frame index. The resulting time-frequency representation is  
\begin{equation*}
R[t,f] = \sum_{i=0}^{\mathcal{I}-1} \bar{\boldsymbol{r}}[\tau_f+i] w[i]
\exp [-j  2\pi \frac{t}{T_{\rm stft}} i ], \forall t \in \mathcal{T}, \forall f \in \mathcal{F},
\end{equation*}
where $w(i)$ denotes the window function with length $\mathcal{I}$. Applying the STFT across all time frames results in the time–frequency matrix $\mathbf{R} \in \mathbb{C}^{T_{\rm stft} \times F}$. Further, we have $\boldsymbol{r} \in \mathbb{C}^{T_{\rm stft}F}$ by vectorizing $ \mathbf{R}$. 

Since all the above operations are linear, $\boldsymbol{r}$ can be expressed in the same form of \eqref{Eq:ReceivedSignalMatchedFilter}. After normalizing each element by $\sqrt{P_{\mathsf r}}$, the echo can be represented as
\begin{equation}\label{Eq:receivedr_k}
    \boldsymbol{r} = \boldsymbol{g} + \boldsymbol{z},
\end{equation}
where $\boldsymbol{g} \in \mathbb{C}^{T_{\rm stft}F}$ represents the ground-truth sensory data of the target, and $\boldsymbol{z} \sim \mathcal{CN}(\boldsymbol{0}, \frac{\sigma_{\mathsf r}^2}{P_{\mathsf r}} \mathbf{I}_{{T_{\rm stft}F}\times {T_{\rm stft}F}})$ denotes the AWGN. 
\subsubsection{Sensing Feature Extraction}
Performing local feature extraction on \eqref{Eq:receivedr_k} using principal component analysis (PCA), the extracted feature vector is given by
\begin{equation}\label{Eq:Observation}
\begin{aligned}
\tilde{\boldsymbol{x}} = \boldsymbol{U}^\top \boldsymbol{r} = \underbrace{\boldsymbol{U}^\top \boldsymbol{g}}_{\text{Ground-truth } \boldsymbol{x}} + \underbrace{\boldsymbol{U}^\top \boldsymbol{z}}_{\text{Sensing noise }\boldsymbol{d}} = \boldsymbol{x} + \boldsymbol{d},
\end{aligned}
\end{equation}
where $\bm{U} \in \mathbb{C}^{T_{\rm stft}F \times M}$ is a unitary matrix, and $\tilde{\boldsymbol{x}} \in \mathbb{C}^M$ denotes the extracted feature vector. It consists of the ground-truth feature component $\bm{x}$ and the transformed sensing noise $\bm{d} = \bm{U}^\top \bm{z}$ \cite{wen2023task}. Since the unitary transformation does not change the variance, we have $\bm{d} \sim \mathcal{CN}(\bm{0}, \frac{\sigma_{\mathsf r}^2}{P_{\mathsf r}} \mathbf{I}_{M\times M})$. For tractability, we model $\bm{x}$ as following multivariate Gaussian distribution.
\begin{assumption}\label{assumption1}
    Considering a classification task with $L$ classes indexed by $\ell \in \mathcal{L} = \{0, \cdots, {L-1}\}$.
    The features of the $\ell$-th class follow a multivariate Gaussian distribution, i.e.,  $p(\bm{x}|\ell) \sim \mathcal{CN}(\bm{x}\,|\,\boldsymbol{\mu}_\ell, \boldsymbol{\Sigma})$ with mean vector $\boldsymbol{\mu}_{\ell}$ and diagonal covariance matrix $\boldsymbol{\Sigma} \in \mathbb{R}^{M \times M}$. Specifically, $\boldsymbol{\mu}_\ell = [\mu_{\ell,0},\mu_{\ell,1},\cdots,\mu_{\ell,M-1}]^\top$ and $\boldsymbol{\Sigma}= \mathrm{diag} (\sigma_0^2,\sigma_1^2,\cdots,\sigma_{M-1}^2)$. Assuming class prior probabilities $\pi_{\ell}$, the overall feature distribution follows a Gaussian mixture (GM) model as
    \begin{equation}
        p(\bm{x}) = \sum_{\ell=0}^{L-1} \pi_{\ell}p(\bm{x}|\ell)\overset{(a)}{=} \frac{1}{L} \sum_{\ell=0}^{L-1} p(\bm{x}|\ell), \forall \ell\in\mathcal{L},
    \end{equation}  where (a) is founded on the uniform priors assumption.
\end{assumption}
Based on the above assumption, we have $\tilde{\boldsymbol{x}} \sim \mathcal{CN}(\bm{0}, \boldsymbol{\Sigma}+\frac{\sigma_{\mathsf r}^2}{P_{\mathsf r}} \mathbf{I}_{M\times M})$. Next, we model the feature transmission channel.
\subsection{Feature Transmission Model}

We consider a block fading channel with a coherence duration of $T_{\rm cd}$ time slots. Under the CE scheme, ISAC device extracts a feature vector 
$\tilde{\boldsymbol{x}} = [\tilde{x}_0, \dots, \tilde{x}_{M-1}]^\top \in \mathbb{C}^{M}$ and transmits it over $N$ parallel subcarriers. We assume that each feature element is transmitted over a subcarrier and $M\leq N$. The received signal at the $n$-th subcarrier is
\begin{equation}
y_{n} = h_n \, b_n \, \tilde{x}_n + w, \quad n=0,\dots,N-1,
\end{equation}
where $h_n \sim \mathcal{CN}(0,1)$ denotes the channel coefficient, $b_n$ is the transmit precoding coefficient, $\tilde{x}_n$ is the transmitted feature, and $w \sim \mathcal{CN}(0, \sigma_w^2)$ represents the AWGN. By collecting all subcarriers, we have ${\boldsymbol y} = \left[{y}_0, {y}_2, \cdots,{y}_{N-1}\right]$. The total transmit power constraint across all subcarriers is
\begin{equation}\label{eq:power_constraint}
\sum_{n=0}^{N-1} |b_n|^2 \, \nu_n^2 \le P_{\mathsf c}, \quad \nu_n^2 \triangleq \mathbb{E}[|\tilde{x}_n|^2],
\end{equation}where the power allocation can leverage prior knowledge of $\nu_n^2$, which can be estimated from the offline training data samples \cite{wen2023task,chen2024view}. 
\begin{equation}\label{eq:processed}
    \hat{x}_n = a_n y_n, \quad
\hat{\boldsymbol{x}} = [\hat{x}_0, \hat{x}_2, \dots, \hat{x}_{N-1}]^\top.
\end{equation}

At the receiver, each subcarrier is processed independently with a scaling coefficient $a_n$, yielding \eqref{eq:processed}.

\begin{remark}
Here, the transmitted signal at each subcarrier is viewed as a continuous signal instead of a discrete one, i.e., discrete-time analog transmission \cite[Remark 2]{shao2021federated}, where baseband symbols are transmitted directly after OFDM modulation, instead of passband digital signals \cite{dong2025robust}. 
\end{remark}

\subsection{Inference Model}
We formulate the classifier as a function mapping the received features into a class label $\mathrm{g}(\cdot): \mathbb{C}^M \rightarrow \mathcal{L}$. To assess the performance of classifier $\mathrm{g}$, we define the classification error as
\begin{equation}
    \!\!\!P_e\!\!  =\!\!\!\sum_{\ell=0}^{L-1} \pi_{\ell} p\left(\mathrm{g}(\hat{\bm{x}} )\!\! \neq\!\! \ell \!\mid \!\ell \right) 
\!\! =\!\!\frac{1}{L}\sum_{\ell=0}^{L-1}  \int I\left(\hat{\bm{x}} ) \neq \ell\right) p(\hat{\bm{x}}|\ell)d \hat{\bm{x}},\! \!\label{Eq:ClassificationError}
\end{equation}
where $I(\cdot)$ is the indicator function yielding 1 if its argument is true and 0 otherwise. Under the uniform prior, the optimal classifier reduces to a maximum likelihood (ML) classifier as
\begin{equation}
\mathrm{g}^\star(\hat{\bm{x}} )=\arg \max _{\ell} p(\hat{\bm{x}} |\ell),\label{eq:QptimalClassifier}
\end{equation}
which minimizes the $P_e$ in \eqref{Eq:ClassificationError}. We can notice the Markov property of the system as ${\bm g}  \rightarrow {\bm r} \rightarrow
	\tilde{\boldsymbol{x}}\rightarrow
	\boldsymbol{y}\rightarrow \hat{\boldsymbol x}\rightarrow \ell.$
\section{Qualification of Inference Accuracy}\label{sec:Theoretical_Analysis}
The inference accuracy depends on $p\left(\mathrm{g}(\hat{\bm{x}} ) \neq \ell \mid \ell \right)$ in \eqref{Eq:ClassificationError}. To optimize the inference accuracy, we should first derive its closed-form expression. However, it is difficult to obtain an explicit expression for the exact error probability. Hence, we resort to using error probability bound to enable analytical tractability.

\begin{theorem}\label{theorem1}
The inference error probability can be lower bounded as
\begin{equation}\label{Eq:unionBound}
    P_e \geq (L-1)Q\Big( \sqrt{\frac{\mathrm{DG}_{\min}}{2}} \Big),
\end{equation}
where $\mathrm{DG}_{\min}$ denotes the minimum DG between any two classes, and $Q(x)=\frac{1}{\sqrt{2 \pi}} \int_x^{\infty} e^{-\frac{t^2}{2}} d t$ denotes the Gaussian $Q$-function. 
\end{theorem}

\begin{proof}
We first consider the simplified scalar binary classification problem, where each class follows a unimodal Gaussian distribution, i.e, $x_\ell \sim \mathcal{CN}(\mu_\ell, \sigma^2), \quad \ell \in \{0,1\}.$ \footnote{Here, we focus on characterizing inference performance and neglect the effects of sensing and communication noise, i.e., we consider $\boldsymbol{x}$ rather than $\tilde{\boldsymbol{x}}$ in \eqref{Eq:Observation}. Since only two classes are involved, the multi-class classification problem reduces to a binary hypothesis testing problem. Moreover, with each class feature being one-dimensional, the communication system reduces to a single-carrier scenario ($N=1$), thus $\sigma_n^2$ reduces to $\sigma^2$.} The Bayes error serves as the lower bound on the error probability \cite[Chapter 3]{fukunaga2013introduction}. Assuming a uniform prior, the minimum error probability achieved by the ML detector is given by  
\begin{equation}\label{The_error_probability_binary_classification}
P_e = Q\!\left(\sqrt{\frac{\mathrm{DG}}{2}}\right), 
\end{equation}
where $\mathrm{DG}$ is the discriminative gain (variance-normalized distance between the two classes) \cite{wen2023task,lan2022progressive,chen2024view,wang2025ultra}, defined as
\begin{equation}\label{eq:scalarcase}
\mathrm{DG} = \frac{(\mu_0 - \mu_1)^2}{\sigma^2}.
\end{equation}
We then generalize the binary classification to the multivariate Gaussian case, i.e., ${\boldsymbol{x}}_{\ell} \sim \mathcal{CN}(\boldsymbol{\mu}_\ell, {\boldsymbol{\Sigma}}),\quad \ell \in \{0,1\}$, as shown in Assumption~\ref{assumption1}. By slight abuse of notation, the log-likelihood ratio test reduces to comparing the DG as
\begin{equation}\label{eq:DG_Matrix}
\!\mathrm{DG}\! = \!(\boldsymbol{\mu}_0 - \boldsymbol{\mu}_1)^\top {\boldsymbol{\Sigma}}^{-1} (\boldsymbol{\mu}_0 - \boldsymbol{\mu}_1) \!\overset{(a)}{=}\! \sum_{m=0}^{M-1}\!\!\frac{(\mu_{0,m}\! -\! \mu_{1,m})^2}{\sigma_m^2},
\end{equation}
where (a) follows from the statistical independence of the feature dimensions after PCA. Since $\eqref{eq:DG_Matrix}$ is the direct generalization of the scalar case \eqref{eq:scalarcase}, the error probability can be expressed as $P_e = Q\!\left(\sqrt{\frac{\mathrm{DG}}{2}}\right)$.

Finally, we further generalize the binary classification with multivariate Gaussian to multiple classes. It is difficult to obtain an explicit expression for the exact error probability. So, we examine the pairwise union bound. Without loss of generality, we can assume $\ell=0$ is transmitted, the union bound says that $P(\hat{\ell}=1|{\ell=0})+P(\hat{\ell}=2|{\ell=0})+\cdots+P(\hat{\ell}=L-1|{\ell=0})$, we have \eqref{Eq:unionBound} \footnote{We note that this result is independently arrived at \cite{wang2025ultra} with a different heuristic proof.}.
\end{proof}
\begin{remark}\label{remark2}
The DG can be interpreted as a generalized form of signal-to-noise ratio (SNR). Class information in an inference task can be viewed as constellation points in digital modulation \cite[Eq. 3.13]{tse2005fundamentals}. Hence, the transmitted features serve as the information-bearing entities. Consider the simplest Rayleigh fading channel model: $y = h x + w$, where $h \in \mathbb{C}$ and $w \sim \mathcal{CN}(0,\sigma_w^2)$. Then, $y \sim \mathcal{CN}(h \mu_\ell, |h|^2 \sigma^2 + \sigma_w^2)$, and the resulting DG at the receiver is $ \frac{|h|^2 |\mu_0 - \mu_1|^2}{|h|^2 \sigma^2 + \sigma_w^2}.$
\end{remark}
Since $Q(\cdot)$ is monotonically decreasing, Theorem 1 implies that increasing $\mathrm{DG}_{\min}$ reduces the lower bound of the misclassification probability. Conventionally, communication and sensing systems are designed with the objective of minimizing the mean squared error (MSE). This naturally raises the question of whether maximizing the DG could lead to improved inference performance.  

\section{Two Cases}\label{sec:Three_Cases}
To gain further insights into the above question, we present two case studies on feature transmission under two design criteria: MSE-optimal and DG-optimal designs.

\subsection{Single-Carrier Case}\label{sec:MSEversusDiscriminativeGain}
We first consider the scalar binary classification case, i.e., ${x}_{\ell} \sim \mathcal{CN}({\mu}_\ell, \sigma^2)$ and $\ell \in \mathcal{L} = \{0,1\}$. The target probability distribution is thus given by $p(x) = \frac{1}{2} \sum_{\ell=0}^{L-1} p(x_\ell|\ell), \quad \forall \ell \in \mathcal{L}$. The exact error probability, given in \eqref{The_error_probability_binary_classification}, is a monotonically decreasing function of the DG. Therefore, we next examine the achievable DG under the MSE-optimal design and then investigate whether direct DG maximization can yield additional performance gains \footnote{Here, we focus on the communication process and neglect the impact of sensing noise $\sigma_{\mathsf r}^2$.}.
\begin{proposition}\label{proposition1} Under the MMSE design criterion and a transmit power constraint $P_{\mathsf c}$, the optimal S\&C transceiver (i.e., transmit precoding and receive scaling coefficients) are
\begin{equation}\label{eq:optimalprecoding}
b^\star = \sqrt{\frac{P_{\mathsf c}}{\nu^2}} \, e^{-j \angle h},\qquad a^\star(b) \;=\; \frac{h\,b\,\nu^2}{\,|h\,b|^2\nu^2+\sigma_w^2\,}, 
\end{equation}
where $\angle h$ symbolizes the phase of the channel. The resulting minimum MSE is: $\mathrm{MSE}^\star = \frac{\nu^2\,\sigma_w^2}{\,|h|^2 P_{\mathsf c} + \sigma_w^2\,}$.
\end{proposition}

\begin{proof}
For a given $b$, the classic MMSE (Wiener) estimator is 
$a^\star(b)=\frac{h\,b\,\nu^2}{|h\,b|^2\nu^2+\sigma_w^2}$ \cite{kay1993fundamentals}, which yields the minimum MSE 
$\mathrm{MSE}_{\min}(b)=\frac{\nu^2\,\sigma_w^2}{|h\,b|^2\nu^2+\sigma_w^2}$.
Since $\mathrm{MSE}_{\min}(b)$ is strictly decreasing in $|b|$, the optimum is achieved by full power transmission and phase-matching under the constraint $|b|^2\nu^2 \le P_{\mathsf c}$, i.e.,
$b^\star = \sqrt{\frac{P_{\mathsf c}}{\nu^2}} \, e^{-j \angle h}$.
Substituting $b^\star$  into $\mathrm{MSE}_{\min}(b)$ yields $\mathrm{MSE}^\star$.
\end{proof}
Using the MSE-optimal precoding and receive scaling coefficients in \eqref{eq:optimalprecoding}, the resulting DG is (see Remark~\ref{remark2})
\begin{equation}\label{eq:discriminativegainMSE}
    \mathrm{DG}^{\star} = \frac{|h|^2 P_{\mathsf c} (\mu_0 - \mu_1)^2/\sigma^2}{|h|^2 P_{\mathsf c} + \sigma_w^2\nu^2/\sigma^2}, 
\end{equation}
\begin{proposition} Using the DG-optimal design, the optimal DG $\mathrm{DG}^{\star}$ is same with \eqref{eq:discriminativegainMSE}.
\end{proposition}

\begin{proof}
The DG of the received signal $\hat x$ in \eqref{eq:processed} can be expressed as
\begin{equation}
    \mathrm{DG} = \frac{|a h b|^2(\mu_0 - \mu_1)^2}{|ah b|^2\sigma^2+a^2\sigma_w^2} = \frac{|h b|^2(\mu_0 - \mu_1)^2}{|h b|^2\sigma^2+\sigma_w^2}.
\end{equation}
It is clear that $\mathrm{DG}$ is monotonically increasing with $|h b|^2$ and upper-bounded by $\mathrm{DG}_{\max}=(\mu_0 - \mu_1)^2/\sigma^2$. Under the transmit power constraint, the achievable effective gain $\mathrm{DG}^{\star}$ is same with \eqref{eq:discriminativegainMSE}.
\end{proof}
The above proposition indicates that directly optimizing DG in the single-feature scenario does not yield any performance improvement over the MSE-optimal design. In this setting, the design of S\&C transceiver depends solely on $h$ and $\mathrm{DG}_{\max}$. To characterize the average system behavior, we take the expectation over $h \sim \mathcal{CN}(0,1)$ to obtain the overall achievable DG.

\begin{theorem}\label{theorem2_proof}
Since the $h$ is drawn independently from $\mathcal{CN}(0,1)$, the average $\mathrm{DG}^{\star}$ is
\begin{equation}\label{eq:average_DG}
     \mathbb{E}[\mathrm{DG}^{\star}]=
\mathrm{DG}_{\max}\left[\,1-\frac{1}{\rho}\,
e^{\frac{1}{\rho}}\,
E_1\!\left(\frac{1}{\rho}\right)\right],
\end{equation}
where $E_1(z)=\displaystyle\int_z^\infty \frac{e^{-t}}{t}\,dt$ denotes the exponential integral and $\rho =\frac{\sigma^2 P_{\mathsf c}}{\sigma_w^2\nu^2}$ denotes the equivalent receive SNR.
\end{theorem}
\begin{proof}
Rewriting \eqref{eq:discriminativegainMSE} as $
\mathrm{DG}^{\star}
=\mathrm{DG}_{\max}\!\big(1-\frac{\frac{\sigma_w^2\nu^2}{\sigma^2}}{|h|^2 P_c+\frac{\sigma_w^2\nu^2}{\sigma^2}}\big)$,
since $|h|^2 \sim \mathrm{Exp}(1)$, we have
$\mathbb{E}[\mathrm{DG}^{\star}]
=\mathrm{DG}_{\max}\!\big(1-\frac{\sigma_w^2\nu^2}{\sigma^2}
\int_{0}^{\infty}\frac{e^{-x}}{P_c x+\frac{\sigma_w^2\nu^2}{\sigma^2}}\,dx\big).$
 $u = x+\frac{\sigma_w^2\nu^2}{\sigma^2 P_c}$ yields
\eqref{eq:average_DG}.
\end{proof}
Theorem~\ref{theorem2_proof} shows that $\mathbb{E}[\mathrm{DG}^{\star}]\propto \mathrm{DG}_{\max}$ when $\rho\to \infty$ and $\mathbb{E}[\mathrm{DG}^{\star}]\propto \rho\mathrm{DG}_{\max}$ when $\rho\to 0$. Intuitively, when feature dimensions are heterogeneous, performance can be further improved by allocating power across these dimensions according to their importance. This suggests that DG-optimal design offers performance gains primarily when multiple feature dimensions ($M > 1$) or parallel time/frequency resources are available, enabling selective power allocation.
\subsection{Multi-Carrier Case}\label{sec:ExtensionMSEversusDiscriminativeGain}
We now formalize the above intuition by extending \emph{MSE-DG analysis} to the case of multi-dimensional features over parallel subcarriers.
\subsubsection{MSE-Optimal design}
For a fixed transmit coefficient $b_n$, the optimal MMSE estimator is identical to Proposition \ref{proposition1} and is given by
\begin{equation*}
    a_n^\star = \frac{h_n b_n \sigma_n^2}{|h_n b_n|^2 \sigma_n^2 + \sigma_w^2}, \quad
\mathrm{MSE}_n = \frac{\sigma_n^2 \sigma_w^2}{|h_n b_n|^2 \sigma_n^2 + \sigma_w^2}.
\end{equation*}
For the MSE-optimal design, the general optimization problem can be formulated as
\begin{equation*}
(\mathrm{P}1) \min_{\{b_n \ge 0\}_{n=0}^{n=N-1}} \quad \sum_{n=0}^{N-1} \mathrm{MSE}_n, \quad
\text{s.t.}  \eqref{eq:power_constraint}.
\end{equation*}
$(\mathrm{P}1)$ is a convex problem and the optimal solution can be
obtained by checking the stationarity of Karush-Kuhn-Tucker (KKT) conditions as \cite{boyd2004convex}
\begin{equation}\label{eq:optimalsolutionmse}
b_n^\mathrm{MSE} = \frac{1}{|h_n|}\sqrt{\left(\frac{\sigma_w|h_n|}{\nu_n\sqrt{\lambda}}-\frac{\sigma_w^2}{\sigma_n^2}\right)^+}e^{-j \angle h_n},
\end{equation}
where $(\cdot)^{+} \triangleq \max\{\cdot, 0\}$ and $\lambda$ denotes the water-filling factor.
\subsubsection{DG-Optimal design}
For the DG-optimal design, the general optimization problem can be formulated as
\begin{align*} (\mathrm{P}2) \max_{\{b_n\}_{n=0}^{n=N-1}} \sum_{n=0}^{N-1}\mathrm{DG}_{n} \quad\text{s.t.}   \eqref{eq:power_constraint}, \end{align*}
where $\mathrm{DG}_{n}={\frac{|h_n b_n|^2\,(\mu_{0,n} - \mu_{1,n})^2}{|h_n b_n|^2\sigma_n^2+\sigma_w^2}}$.
 $(\mathrm{P}2)$ remains convex and and its optimal solution is
\begin{equation}\label{eq:optimalsolutiondg}
    \!b_n^\mathrm{DG} =  
\dfrac{1}{|h_n| } \sqrt{\left( \dfrac{\sigma_w|h_n| |\mu_{0,n} - \mu_{1,n}| }{\nu_n\sqrt{\lambda} \sigma_n^2} - \frac{\sigma_w^2}{\sigma_n^2} \right)^+} e^{-j \angle h_n},
\end{equation}
\subsubsection{Solution Structure Analysis}
Based on \eqref{eq:optimalsolutionmse} and \eqref{eq:optimalsolutiondg}, we have the following three insights, which are also numerically validated in Fig. \ref{Fig_MSE_optvsDG_opt}.
\begin{itemize}
    \item \textbf{Structural difference:} \eqref{eq:optimalsolutiondg} differs from that in \eqref{eq:optimalsolutionmse} by an additional term $\frac{|\mu_{0,n} - \mu_{1,n}|}{\sigma_n^2}$.
    \item \textbf{High SNR Case:} When $\sigma_w^2$ is small and $|h_n|$ is large, we have $b_n^{\mathrm{MSE}} \approx \frac{1}{|h_n|} \sqrt{\frac{\sigma_w |h_n|}{\nu_n \sqrt{\lambda}}} = \sqrt{\frac{\sigma_w}{|h_n| \, \nu_n \sqrt{\lambda}}}$ and $b_n^{\mathrm{DG}} \approx \sqrt{\frac{\sigma_w \, }{|h_n| \, \nu_n \sqrt{\lambda}}\frac{|\mu_{0,n}-\mu_{1,n}|}{\sigma_n^2}}\approx\sqrt{\frac{\sigma_w}{|h_n| \, \nu_n \sqrt{\lambda}}}.$ Both allocation strategies are approximately consistent.
    \item \textbf{Low SNR Case:} When the term $\sigma_w^2 / \sigma_n^2$ dominates, weak features may be completely turned off, i.e.,$
\frac{\sigma_w |h_n| |\mu_{0,n}-\mu_{1,n}|}{\nu_n \sqrt{\lambda} \, \sigma_n^2} - \frac{\sigma_w^2}{\sigma_n^2} < 0$ and thus $b_n^\mathrm{DG} = 0$. This results in a selective power allocation where only highly discriminative or well-conditioned features receive power. In contrast, the MSE-optimal design tends to allocate power more uniformly, potentially wasting resources on uninformative features.
\end{itemize}



\begin{figure}[!h]
    \centering
    \includegraphics[width=0.4\textwidth]{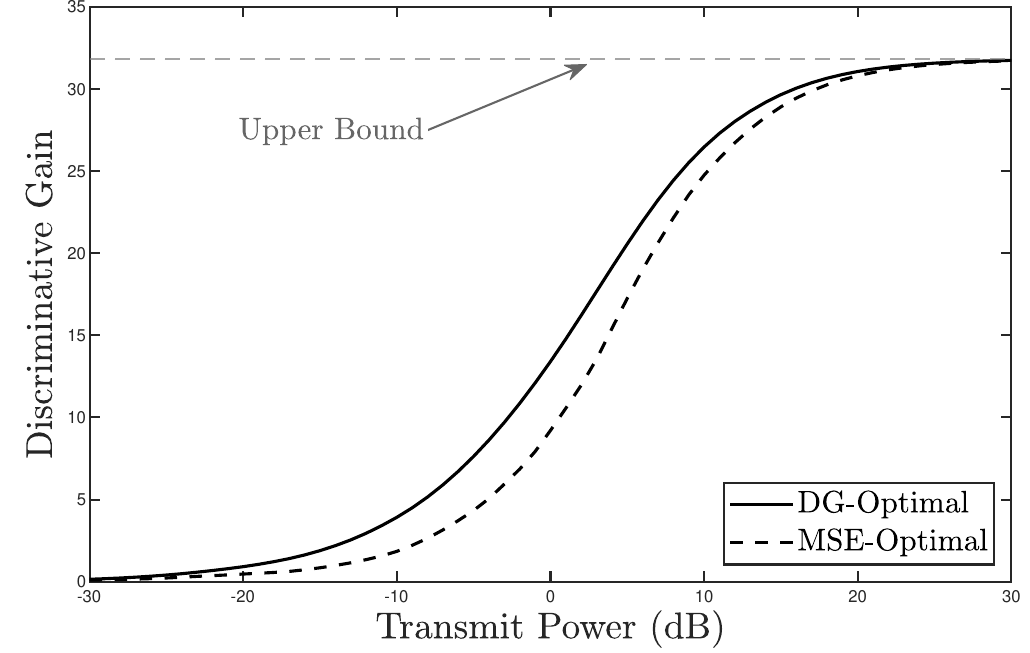}
    \caption{Discriminative gain versus transmit power under DG-optimal and MSE-optimal designs.}
    \label{Fig_MSE_optvsDG_opt}
\end{figure}

\subsubsection{Saving Transmit Power for Sensing}
 Fig. \ref{Fig_MSE_optvsDG_opt} shows that DG-optimal transmission consistently outperforms MSE-optimal transmission in terms of power efficiency. This result highlights an important insight: inference-oriented communication inherently requires less transmit power to achieve a target inference accuracy compared with MSE-based design. Consequently, the saved power can be reallocated to the sensing, as indicated by the transmit signal model in \eqref{Eq:signalTransmitted}. 

\section{Optimized S\&C Power Allocation}\label{sec:Task-Oriented_Feature_Transmission}
Building on the results of the previous sections, where inference-oriented feature transmission could save communication power and thus reduce the power requirement for sensing, we now proceed to jointly optimize the sensing and communication power allocation, explicitly taking into account the impact of sensing noise. The optimization problem is formulated as
\begin{align*} (\mathrm{P}5) \max_{P_{\mathsf r},P_{\mathsf c}} \sum_{n=0}^{N-1}\widetilde{\mathrm{DG}}_{n} \quad\text{s.t.}   \eqref{Eq:P_r_P_c},\eqref{eq:power_constraint},\end{align*}
where $\widetilde{\mathrm{DG}}_{n}={\frac{|h_n b_n|^2\,(\mu_{0,n} - \mu_{1,n})^2}{|h_n b_n|^2\tilde{\sigma}_n^2+\sigma_w^2}}$ and $\tilde{\sigma}_n$ denotes the effective per-class variance with sensing noise, i.e., $\tilde{\sigma}_n = \sigma_n + \frac{\sigma_{\mathsf r}^2}{P_{\mathsf r}},$ as defined in \eqref{Eq:Observation}. 

For FMCW sensing, maximum DG is equivalent to maximizing the effective sensing power. Hence, the optimal solution $\{P_{\mathsf r}, P_{\mathsf c}\}$ is achieved if and only if the equality holds in the power constraints \eqref{Eq:P_r_P_c}. Since the optimal carrier structure is already characterized by \eqref{eq:optimalsolutiondg}, the optimal solution $\{P_{\mathsf r}, P_{\mathsf c}\}$ can be obtained by simply searching over the feasible power allocation pair $(P_{\mathsf r}, P_{\mathsf c})$ rather than solving problem $(\mathrm{P}5)$ directly. This reduces the problem to a one-dimensional search along the line $\beta P_{\mathsf c} + (1-\beta) P_{\mathsf r} = P,$ where $\beta \in [0,1]$ is the allocation ratio.

\section{Numerical Results} \label{sec:results}
In this section, numerical results are provided to validate the analytical results. 
\subsection{Settings}
We adopt multilayer perceptron (MLP) and support vector machine (SVM) as the inference models. Both are trained on noise-free PCA features extracted from a human posture recognition dataset \cite{wen2023task}. The dataset contains 12,000 samples, which are split into training, validation, and test sets in a 3:1:1 ratio. The task is formulated as a four-class human posture classification problem, where each sample corresponds to a time–frequency spectrogram obtained via FMCW sensing. Key system parameters are set as follows: $\sigma_{\mathsf r}^2 = \sigma_w^2 = 0.1$, $P \in [-5,10]$ dB, and $M = N = 8$.

\subsection{Results}
We first re-examine Fig.~\ref{Fig_MSE_optvsDG_opt} across different inference models to validate the preceding observations (see Fig.~\ref{Fig_MSE_optvsDG_opt_inference}). Then, the trade-off between communication and sensing power under inference performance maximization is revealed (see Fig.~\ref{Fig_results_compare_P_r_vs_P_c}). 
\begin{figure}[h]
	\centering
        \begin{minipage}{0.22\textwidth}
		{\includegraphics[width=\textwidth]{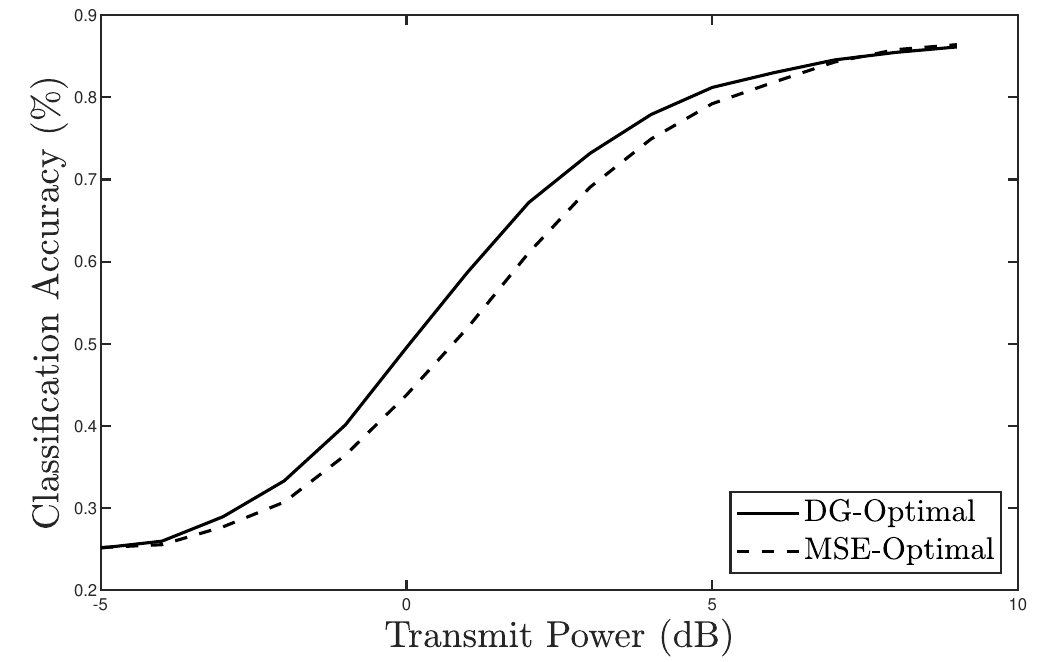}} 
            \caption*{(a)}\label{Fig_mlp}
	\end{minipage}
	\begin{minipage}{0.22\textwidth}
		{\includegraphics[width=\textwidth]{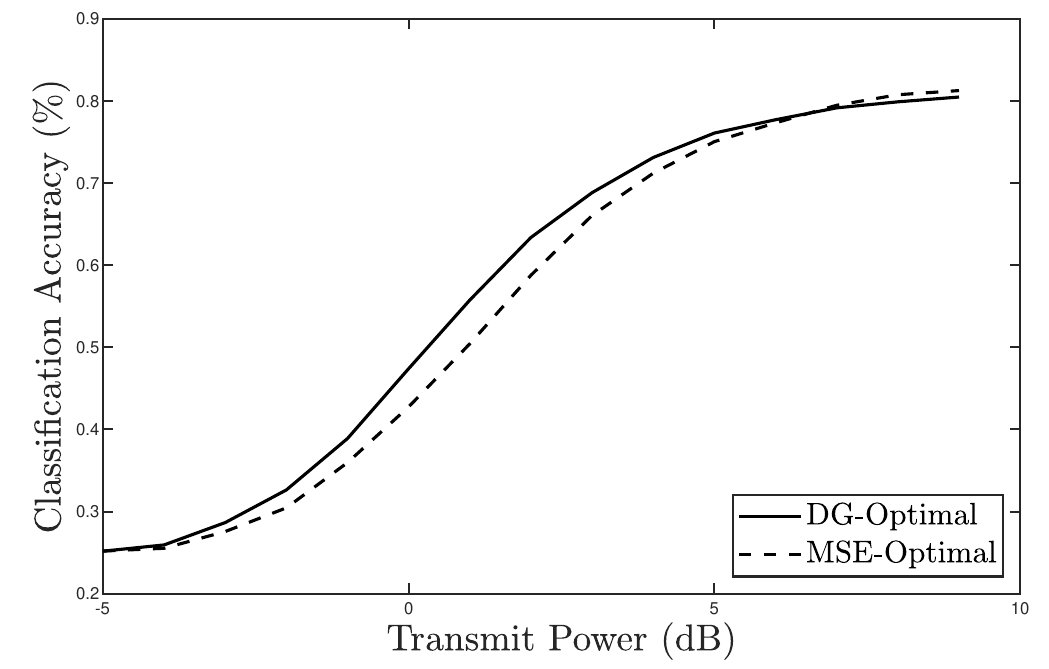}}
		\caption*{(b)}\label{Fig_svm}
	\end{minipage}
\caption{(a) Inference performance versus communication power $P_{\mathsf c}$ under the MLP model. (b) Inference performance versus communication power $P_{\mathsf c}$ under the SVM model. }\label{Fig_MSE_optvsDG_opt_inference}
\end{figure}

\begin{figure}[!h]
    \centering
    \includegraphics[width=0.475\textwidth]{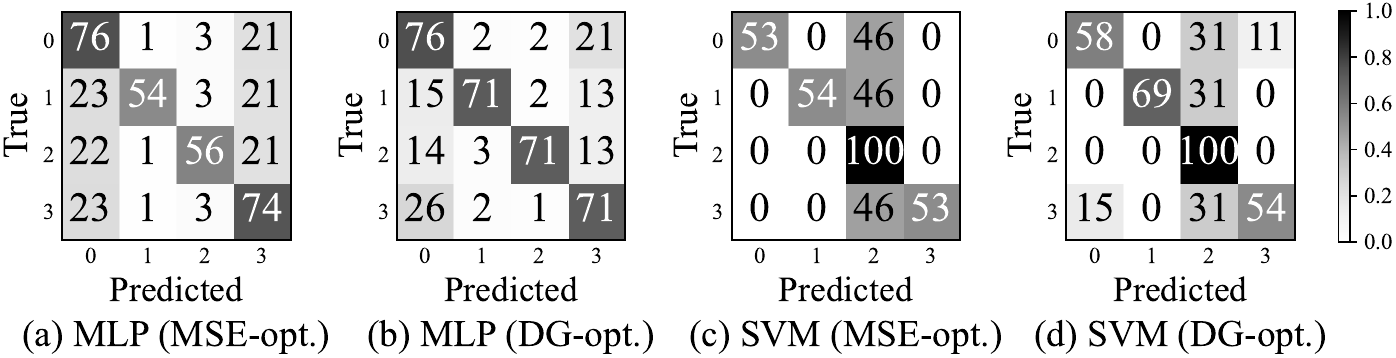}
    \caption{Confusion matrices under $P_{\mathsf c}=2$ dB for different inference models and optimality criteria.}
    \label{fig_CM_2x2_P2dB_K200}
\end{figure}

In Fig.~\ref{Fig_MSE_optvsDG_opt_inference}, both the SVM and MLP models exhibit trends consistent with Fig.~\ref{Fig_MSE_optvsDG_opt}. In the low-SNR regime, power is selectively allocated to features with strong DG or favorable channel conditions, resulting in superior inference performance compared with MSE-optimal designs. When transmit power is very low, inference accuracy approaches random guessing at $1/L$, whereas at very high power, performance saturates due to the intrinsic class separability and the network capacity. Moreover, Fig.~\ref{fig_CM_2x2_P2dB_K200} shows that DG-optimal design yields a more balanced classification accuracy across all categories, highlighting the benefit of inference-oriented optimization.
\begin{figure}[!h]
	\centering
        \begin{minipage}{0.4\textwidth}
		{\includegraphics[width=\textwidth]
         {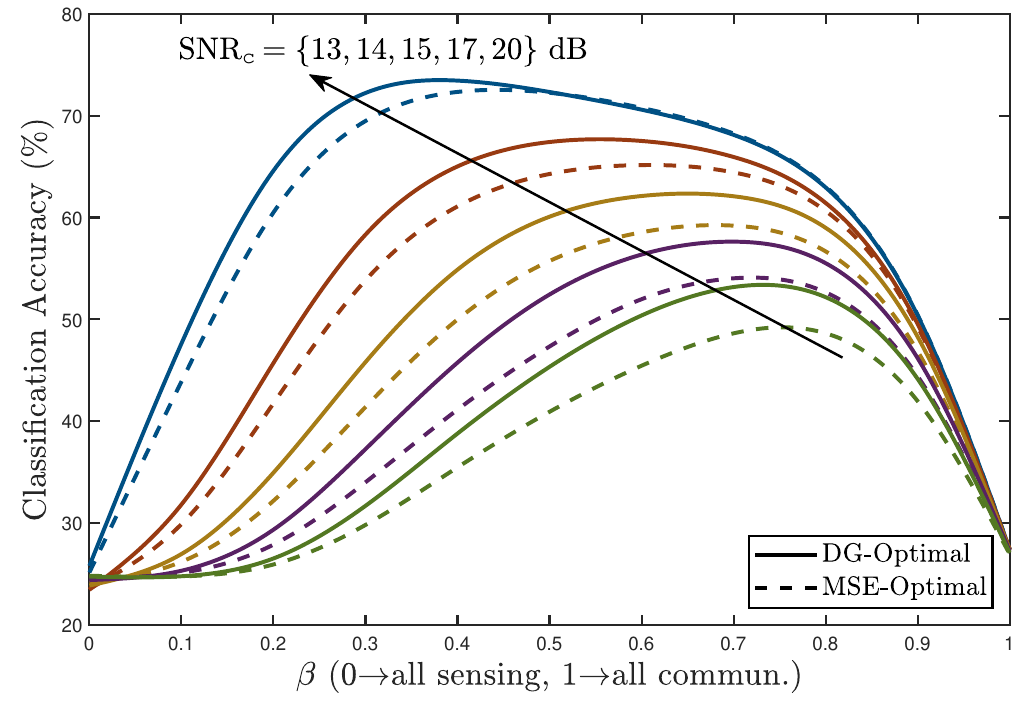}}
		\caption*{(a)}\label{Fig_L}
	\end{minipage}
            \begin{minipage}{0.41\textwidth}
		{\includegraphics[width=\textwidth]
         {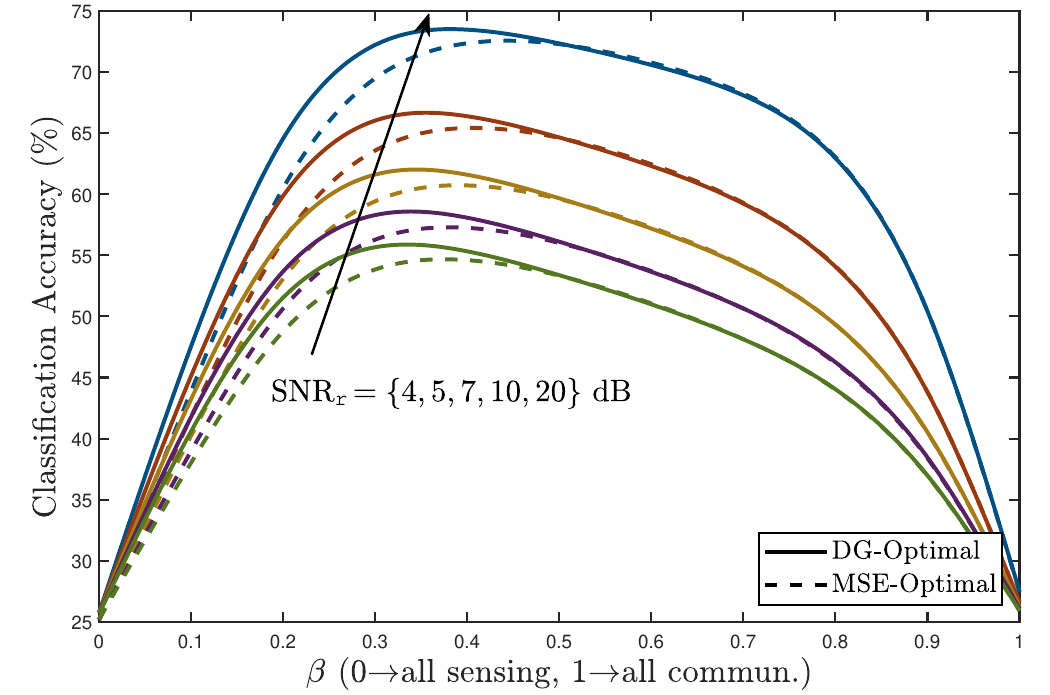}}
		\caption*{(b)}\label{Fig_L}
	\end{minipage}
\caption{(a) Inference performance versus allocation ratio $\beta$ under different $\mathrm{SNR}_{\mathsf c}$. (b) Inference performance versus allocation ratio $\beta$ under different $\mathrm{SNR}_{\mathsf r}$. }\label{Fig_results_compare_P_r_vs_P_c}
\end{figure}

In Fig.~\ref{Fig_results_compare_P_r_vs_P_c}(a), the sensing SNR $\mathrm{SNR}_{\mathsf r}$ is fixed at 20 dB while the communication SNR $\mathrm{SNR}_{\mathsf c}$ is varied. Conversely, Fig.~\ref{Fig_results_compare_P_r_vs_P_c}(b) fixes $\mathrm{SNR}_{\mathsf c}$ and varies $\mathrm{SNR}_{\mathsf r}$. Two key observations can be made from Fig.~\ref{Fig_results_compare_P_r_vs_P_c}(a):
(i) as $\mathrm{SNR}_{\mathsf c}$ increases, the power fraction allocated to $P_{\mathsf c}$ decreases because the marginal communication gain saturates, making additional power more valuable for sensing;
(ii) the performance gap between DG-optimal and MSE-optimal designs narrows as $\mathrm{SNR}_{\mathsf c}$ grows, since the system becomes less power-limited and the MSE solution approaches the DG-optimal one. Fig.~\ref{Fig_results_compare_P_r_vs_P_c}(b) shows a similar trend when $\mathrm{SNR}_{\mathsf r}$ is varied, confirming the robustness of this trade-off.

\section{Conclusion} \label{sec:conclusion}
This paper studied the coordination gain of ISAC under a CE framework. Closed-form solutions for DG-optimal and MSE-optimal transceiver designs were derived, revealing water-filling-type structures and explicit S\&C tradeoffs. The developed results suggested that DG-optimal design achieves more power-efficient transmission, particularly in the low-SNR regime, by selectively allocating power to key features and saving power for sensing. Numerical experiments validated the theoretical findings and confirmed consistent improvements in inference accuracy and power efficiency over MSE-optimal schemes.


\bibliographystyle{IEEEtran}
\bibliography{reference/mybib}

\end{document}